\documentclass[sn-mathphys-num]{sn-jnl}% Math and Physical Sciences Numbered Reference Style 
\usepackage{graphicx}%
\usepackage{multirow}%
\usepackage{amsmath,amssymb,amsfonts}%
\usepackage{amsthm}%
\usepackage{mathrsfs}%
\usepackage[title]{appendix}%
\usepackage{xcolor}%
\usepackage{textcomp}%
\usepackage{manyfoot}%
\usepackage{booktabs}%
\usepackage{algorithm}%
\usepackage{algorithmicx}%
\usepackage{algpseudocode}%
\usepackage{listings}%
\usepackage{marginnote}
\usepackage{mathtools}
\usepackage{scalerel}
\usepackage{comment}
\usepackage{dsfont}
\usepackage{tikz}
\usepackage{tikz-cd}
\theoremstyle{thmstyleone}%
\newtheorem{theorem}{Theorem}%  meant for continuous numbers
\newtheorem{proposition}[theorem]{Proposition}% 

\theoremstyle{thmstyletwo}%
\newtheorem{remark}{Remark}%

\theoremstyle{thmstylethree}%
\newtheorem{definition}{Definition}%

\def\CA{{\mathcal A}}
\def\CB{{\mathcal B}}
\def\CM {{\mathcal M}}

\def\CK {{\mathcal K}}
\def\CP {{\mathcal P }}
\def\CU {{\mathcal U}}
\def\CV {{\mathcal V}}
\def\CH {{\mathcal H}}
\def\EE{\mathbb{E}}
\def\K {{\operatorname{K}}}
\def\B {{\operatorname{B}}}
\def\Tr {{\operatorname{Tr}}}
\def\WCM {{\widehat{\mathcal M}}}
\def\CF {{\mathcal F}}

\def\IC{\mathbb{C}}
\def\CO {{\mathcal O}}

\def\IR{{\mathbb{R}}}

\def\IZ{{\mathbb{Z}}}
\def\T{{\mathsf{T}}}

\def\HS{\operatorname{HS}}
\def\bP{\mathbb{P}}
\def\Cs{C$^*$}

\def\a{{\alpha}}
\def\b{{\beta}}
\def\g{{\gamma}}

\def\1{\mathds{1}}

\def\ra{\rangle}
\begin{document}

\title[Article Title]{Parametrized topological phases in 1d and T-duality}

\author{\fnm{Roman} \sur{Geiko}

\affil{\orgdiv{Department of Physics and Astronomy}, \orgname{University of California},

\city{Los Angeles}, \postcode{90095}, \state{CA}, \country{USA}}}

\abstract{There are families of physical systems that cannot be adiabatically evolved to the trivial system uniformly across the parameter space, even if each system in the family belongs to the trivial phase. The obstruction is measured by higher Berry class. We analyze families of topological systems in 1+1d using families of invertible TQFTs and  families of RG fixed states of spin chains. We use the generalized matrix-product states to describe RG fixed points of all translation invariant pure splits states on spin chains. Families of such fixed points correspond to bundles of Hilbert-Schmidt operators. There exists a global MPS parametrization of the family if and only if the latter bundle is trivial. We propose a novel duality of parametrized topological phases which is an avatar of the T-duality in string theory. The duality relates families with different parameter spaces and different higher Berry classes. Mathematically, the T-duality is realized by gauging the circle action on the continuous trace algebra generated by parametrized matrix-product tensors.}


\maketitle

\section{Introduction and conclusion}

Topological phases of one-dimensional quantum lattice systems are well understood by now. All plain bosonic phases are trivial, while there are non-trivial phases protected by a local symmetry group and non-trivial families of parametrized phases. In particular, phases parametrized by a space $\CM$ are classified by a cohomology class $H\in H^3(\CM,\IZ)$, higher Berry class. When $\CM$ is a total space of a principal $\operatorname{U}(1)$-bundle,  we can perform a Fourier transform along the circle. As a result, a family parametrized by $\CM$ with the Berry class $H$ is dual to a family parametrized by another space $\WCM$ and another class $\widehat{H}\in H^3(\WCM,\IZ)$, and vice versa. This paper is intended to serve as an invitation to the subject of T-duality in parametrized topological phases of matter. We review and revisit the material concerning families of 1d phases in order to interpret the T-duality as a transformation of the data parametrizing families of RG fixed states.

This paper is devoted to a study of topological phases of physical systems in one spatial dimension. Primarily, we focus on applications to quantum lattice systems in 1d -- spin chains. Topological phases are represented by sufficiently local Hamiltonians with a unique gapped ground state. Different topological phases correspond to different connected components of the space of such Hamiltonians with respect to the adiabatic evolution. An alternative point of view at topological phases, the one taken in this paper, as of the patterns of entanglement of the unique ground state.

We consider only translation invariant states -- they admit a description in terms of the non-commutative Markov chains \cite{accardi1974}/generalized matrix-product states \cite{bratteli1996endomorphisms,bratteli2000pure,Matsui2001}. This formalism uses auxiliary algebras acting on the virtual/bond space to encode the multi-partite entanglement. In the case of finite-dimensional auxiliary algebra, this is known as the finitely-correlated states \cite{Fannes1992} and matrix-product states \cite{perez2006matrix,cirac2021matrix}. Pure translationally invariant states in 1d in infinite volume with relatively low entanglement (by which we mean the states satisfying the split property in the sense of \cite{Matsui2001}) are completely described the generalized MPS \cite{bratteli1996endomorphisms,bratteli2000pure,Matsui2001} with the virtual von Neumann algebra of type I.  

Generalized MPS (gMPS) allow a uniform description of renormalization fixed points of all TI pure split states. We use the RG fixed states as representatives of the trivial topological phase. Then, we describe families of topological phases as families of invertible TQFTs and families of RG fixed states. We explain how the gluing data for a family of invertible open-closed 2d TQFTs is equivalent to a $\operatorname{U}(1)$ bundle gerbe. Families of RG fixed gMPS tensors are encoded by bundles of Hilbert-Schmidt operators: locally, a field of RG fixed gMPS tensors is a frame for such a bundle. The isomorphism classes of $\operatorname{U}(1)$ bundle gerbes and of algebras of Hilbert-Schmidt operators are given by $H^3(\CM,\IZ)$.
There exists a global gMPS parametrization of a family of states iff the corresponding Hilbert-Schmidt bundle is trivial.

\bmhead{T-duality} 
The T-duality is a duality in string theory and corresponding low-energy non-linear sigma models \cite{PhysRevLett.58.1597} relating theories with different Target-spaces. At the level of NLSMs, the T-duality is a gauging of a $\operatorname{U}(1)$ isometry of the target-space \cite{BUSCHER1988466}. In the context of parametrized phases, the NLSMs appear as effective field theories for the mass parameters for a family of gapped theories \cite{abanov2000theta,cordova2020anomalies,hsin2020berry}. In particular, the non-trivial topological term in an NLSM is interpreted as  higher Berry curvature. The T-dual sigma-model has a different target space with a different metric, B-field, and generally different topology \cite{Bouwknegt_2004}. Here we employ a topological approach to the T-duality which does not involve a metric and which can be applied to topological theories \cite{Bouwknegt_2004,bunke2005topology}.

In the algebraic-geometric duality, gauging a group action on a space is equivalent to gauging the action on an algebra. From the point of view of operator algebras, the T-duality is a transformation of a \Cs-algebra generated by RG fixed parametrized gMPS tensors. As basic examples, we have the following T-dual pairs:
\begin{align}\label{Examples}
    (S^3,1)\longleftrightarrow (S^3,1)\,,\quad (S^2\times S^1,1)\longleftrightarrow (S^3,0)\,,\quad (S^2\times S^1,0)\longleftrightarrow (S^2\times S^1,0)\,
\end{align}
Here ``$0$" and ``$1$''' stand for the higher Berry numbers, i.e., integrals of higher Berry classes over the corresponding three-manifolds. 

\bmhead{Concluding remarks} From the point of view of topological phases, the parameter space and a local $G$-symmetry should seen as two faces of the same coin. This is most transparent in the language of $\operatorname{U}(1)$ extensions of groupoids as was done in \cite{moore2006dbranes}: the parameter space gives us the \v{C}ech groupoid while any group is a groupoid with a single element. Gauging of a finite subgroup of a symmetry group in any dimension was discussed in many places, e.g., \cite{tachikawa2020gauging,Gaiotto_2021}. Here we provide a simple example of gauging in the parameter space.

We expect that T-dual theories can be separated by ``bi-branes'', the topological defects separating two parametrized families, mathematically given by gerbe bimodules \cite{FUCHS2008576}. 

The Abelian T-duality reviewed in this work admits many generalizations such as the G-equivariant \cite{Dove:2023pqy} and non-Abelian and non-commutative T-duality \cite{daenzer2007groupoidapproachnoncommutativetduality}.   

\bmhead{Notation} We use $\CV$, $\CP$, and $\CH$ to denote separable complex Hilbert spaces -- $\CV$ for the virtual/bond space and $\CP$ for physical, and $\CH$ is reserved for the infinite-dimensional space. $\T$, $\HS$, $\K$, and $\B$ stand for the trace-class, Hilbert-Schmidt, compact, and bounded operators respectively. By $\mathbb{P}$ we denote the projectivization of a Hilbert space.

\section{Parametrized families of gapped systems in 1+1d}

In this section we discuss the generalities of parametrized  invertible phases in 1+1d by analyzing families of invertible 2d TQFTs.

Topological phases in 1+1d protected by a local symmetry group $G$ are fully captured by $G$-equivariant 2d TQFTs, in the sense that all existing classifications, e.g., by using the MPS argument \cite{Pollmann2010,chen2013symmetry} or operator algebra methods \cite{ogata2019classification,kapustin2021classification,ogata2021classification}, agree with the classification by 2d TQFTs \cite{Shiozaki:2016,Kapustin_2017}. On the other hand, parametrized topological phases are argued to be classified by $H^3(\CM,\IZ)$, see \cite{BerryPhase2020} for the effective field theory approach, \cite{kapustin2022local,artymowicz2024quantization} for operator algebras approach, and \cite{ohyama2024discrete,ohyama2024higher,qi2023charting} for MPS approach. 

In what follows we give a simple argument for why the gluing of invertible open-closed 2d TQFTs produces a $\operatorname{U}(1)$ bundle gerbe, an object classified by $H^3(\CM,\IZ)$. We do not impose any symmetry, but the analysis of parametrized and $G$-symmetric phases can be conducted by analogy. 

Any closed 2d TQFT $\CF$, as a functor, is completely determined by its value on the circle $\CF(S^1)$, which is a commutative Frobenius algebra. Moreover, there is an equivalence of 1-categories of closed 2d TQFTs and commutative Frobenius algebras \cite{kock2004frobenius}. A 2d TQFT representing the trivial phase is simply the algebra $A\cong \IC$, equipped with a non-zero complex number $\lambda$, the value of the Frobenius form. As we have in mind TQFTs effectively describing spin chain models, we choose the canonical normalization $\lambda=1$ \cite{Shiozaki:2016,Kapustin_2017}; from now on we assume the Frobenius form is canonically normalized.

In order to define an open-closed 2d TQFT, we need an extra datum, the category of boundary conditions \cite{moore2006dbranes}, what $\CF$ assigns to a point. If  $\CF(S^1)=\IC$, then this category is simply the category of complex vector spaces $\mbox{Vec}_{\IC}$, the category of representations of $\IC$. In other words, a trivial phase is characterized by a pair $A=\IC$ and $\mbox{Vec}_{\IC}$. Now, we consider families of such pairs. We glue algebras using algebra automorphisms and their categories of representations using Morita bimodules. Alternatively, we can start with a fully-extended 2-1-0 TQFT with the target 2-category of finite-dimensional $\IC$-algebras, bimodules, and intertwiners. Then, forming a stack over the parameter space is a natural thing to do.

We assume that the parameter space is a smooth compact manifold $\CM$, though many results in this paper can be restated for topological spaces. Let $\CU$ be some good cover of $\CM$, i.e., a finite collection of open sets $\CU=\{\CU_{\a}\}_{\a\in I}$ which are all contractible and all their intersections are contractible -- such a cover always exists for a smooth manifold, see p. 42 in \cite{bott2013differential}. To each $\CU_{\a}$ we associate a copy of $\IC$. Note that we consider $\IC$ as a $\IC$-algebra -- for that reason we associate a copy of $\IC$ to the whole patch instead of forming, say, a line bundle. To each double intersection $\CU_{\a\b}\coloneqq \CU_{\a}\cap \CU_{\b}$, we associate a (trivial, since all intersections are contractible) line bundle $L_{\a\b}$ of $\IC-\IC$ bimodules. The tensor product of a pair of $\IC$, as of bimodules, is isomorphic to $\IC$, with the isomorphism being a non-zero complex number. In other words, to any triple intersection $\CU_{\a\b\g}$, we associate an isomorphism of line bundles
\begin{align}
    \lambda_{\a\b\g}:L_{\a\b}\otimes L_{\b\g}\xrightarrow{\sim}L_{\a\g}\,.
\end{align} 
The associativity of the tensor product is a condition on $\lambda$ holding for any quadruple intersection. In other words, $\lambda$ is a \v{C}ech 2-cocyle valued in the sheaf of $\IC^{\times}$-valued functions, i.e., an element of $Z^2(\CU,\underline{\IC}^{\times})$. The data consisting of the line bundles $\{L_{\a\b}\}_{\a,\b\in I}$, and isomorphisms $\{\lambda_{\a\b\g}\}_{\a,\b,\g \in I}$ is nothing but a data of a $\operatorname{U}(1)$ bundle gerbe in the sense of Hitchin-Chatterjee \cite{Hitchin:1999fh}.

There are natural ways of defining isomorphisms of bundle gerbes, that we are not going to spell out. We recommend \cite{kristel20212} for a thorough discussion of various 2-stacks and where bundle gerbes are presented as special types of 2-line bundles and, more generally, of 2-vector bundles.
In the appropriate category, bundle gerbes defined on different good covers of $\CM$ are equivalent, as bundle gerbes defined on different covers can be compared by pulling them back to the intersection of the two covers. Instead of working with good cover $\CU$, we can take its \v{C}ech nerve $N(\CU)$ so that $H^2(N(\CU),\underline{\IC}^{\times})\cong H^3(N(\CU),\IZ)\cong H^3(\CM,\IZ)$.

\begin{remark}
    Let us conclude with a remark regarding the connection between our parameter space and the ``space-time" from \cite{moore2006dbranes}. Using the algebro-geometric intuition, we can naively imagine a family of Frobenius algebras as a commutative Frobenius algebra whose spectrum ``is'' $\CM$. However, any Frobenius algebra is finite dimensional and has a finite spectrum. This was observed in \cite{moore2006dbranes} and the spectrum was taken as the space-time. According to the discussion above, their space-time is a \v{C}ech nerve of our parameter space $\CM$. Nevertheless, we will not abandon the idea of having an algebra with spectrum $\CM$; we will argue that the appropriate algebra is a continuous-trace \Cs-algebra with spectrum $\CM$.
\end{remark} 
 
\section{Topological states in 1d}

In this section we describe examples of parametrized states on spin chains with zero correlation length. We consider a family of states and assume that each state in the family is translation invariant (TI), pure, and split. Such states are very well studied even in infinite volume thanks to the formalism of generalized matrix-product states \cite{accardi1974,Fannes1992,bratteli1996endomorphisms}. We review the description of such states as well their RG fixed points. In the end, we construct families of the RG fixed points. 

\begin{remark}
Let us comment on the split property. One should think of split states in 1d as having relatively low entanglement, see \cite{keyl2006entanglement} for a review of entanglement types of spin chains. Importantly, a unique ground states of any local gapped Hamiltonian is pure and split, see Corollary 3.2 of \cite{Matsui2013}. Strictly speaking, the split states are more general than the short-range entangled states of \cite{kapustin2021classification}. We stress that the split property is the characterization of the entanglement pattern of the spin-chain rather than its dynamical properties. In particular, correlation functions of local operators in a split state might not decay exponentially fast with the distance.
\end{remark}

The philosophy behind the infinite-volume states is that we define a state on the whole lattice $\IZ$, yet we substitute only (quasi-) local operators. This way, we do not have to worry about existence of the thermodynamic limit. Let $\CP$ be the local Hilbert space of the spin chain. Unless specified, $\CP$ is finitely dimensional. The algebra of operators acting on each site is the algebra of all bounded operators $\B\CP$ (matrix algebra).  Let $\CA_{\IZ}$ be the C$^*$ algebra of quasi-local operators on $\IZ$ which is the operator norm completion of the infinite tensor product of $\B\CP$: 
\begin{align}
    \CA_{\IZ}=\overline{\bigotimes _{\IZ} \B\CP}^{C^*}\,.
\end{align}

\begin{definition}\label{def:SFCS}
Let $\B\CV$ be the von Neumann algebra of all bounded operators on a separable Hilbert space $\CV$. A semi-finitely correlated state (SFCS) on $\CA_{\IZ}$, is a pair $(\EE,\phi)$ consisting of 
\begin{itemize}
    \item A normal completely-positive map $\EE:\B\CP\otimes \B\CV\to \B\CV$ such that $\EE(\1\otimes \_\!\_\, )\coloneqq \Phi(\_\!\_\,)$ is a normal CP map $\B\CV\to\B\CV$, which is also unital $\Phi(\1)=\1$. The identity operator $\1$ is the unique fixed point of $\Phi$.
\item A normal faithful state $\phi$ on $\B\CV$ such that $\phi\circ \Phi(x)=\phi(x)$ for all $x\in \B\CV$. As a normal state on $\B\CV$, it can be represented by a positive trace-class operator $\rho$ of unit trace: $\phi(x)=\Tr[\,\rho\, x\,]$. Moreover, $\rho$ is invertible as $\phi$ is faithful. 
\end{itemize}
\end{definition}
The definitions used here are standard, see, e.g., \cite{Blackadar2006}. An SFCS defines a translation invariant state $\omega$ on $\CA_{\IZ}$ in the following way. Let $\CO=\CO_1\otimes \ldots \otimes\CO_{n}\in \CA_{\{1,\ldots, n\}}\subset \CA_{\IZ}$ be a simple local operator supported on the sites $\{1,\ldots, n\}\subset \IZ$. As a part of $\CA_{\IZ}$, this operator is understood as
\begin{align}
    \ldots \1\ldots \1 \,\CO_1\ldots \CO_{n}\,\1 \ldots \1 \ldots
\end{align}
The expectation value of this operator is given by 
\begin{align}
    \omega(\CO)=\phi[\ldots \Phi (\EE(\CO_1\otimes \EE(\CO_2\ldots  \EE({\CO_n}\otimes \Phi(\ldots ))]=\phi[\EE(\CO_1\otimes \EE(\CO_2\ldots  \EE({\CO_n}\otimes \1)]
\end{align} -- in this equality we used that  $\Phi$ preserves $\phi$. This way, we have a translation invariant state defined everywhere on $\CA_{\IZ}$. The pair $(\Phi,\phi)$ of a normal unital CP map on $\B\CV$ preserving a normal faithful state are known in the literature as W$^*$-dynamical systems \cite{Robinson}. As we are going to see, for pure states, the pairs $(\EE,\phi)$ and $(\Phi,\phi)$ are interchangeable. 

Given a normal unital CP map $\Phi:\B\CV\to \B\CV$, there are equivalent formulations of the Kraus theorem \cite{kraus71}:
\begin{itemize}
    \item There exists a Hilbert space $\CP$ and an isometry 
    \begin{align}\label{Popescu isometry}
        V:\CV\to \CP\otimes \CV\,,\quad \mbox{such that}\quad \Phi(x)=V^*(\1\otimes x)V\,.
    \end{align}
    
    \item There exists a linear map $T:\CP\to \B\CV$, called \textit{generalized MPS tensor}, such that for a choice of basis $\{|i\ra\}_{i=1}^{d}$ for $\CP$, $\Phi$ can be presented by the set of Kraus operators $\{T_i=T(|i\ra)\}_{i=1}^{d}$:
\begin{align}
    \Phi(x)=\sum_{i=1}^{d}T_i\,x\,T_i^{*}\,,\quad x\in \B\CV\,.
\end{align}
\end{itemize}
 A part of the Kraus theorem states that the set of Kraus operators is unique up to a choice of basis for $\CP$. The correspondence between certain normal CP maps and generalized matrix-product tensors can be obtained functorially for the general vN algebras; it is a subject of an upcoming paper \cite{GMM}.
 
 Given the isometry $V$ as above, we can define a CP map $\EE_V:\B\CP\otimes \B\CV\to \B\CV$ via
\begin{align}\label{eq:purelygenr}
   \EE_V: \CO\otimes x\mapsto V^*(\CO\otimes x)V
\end{align}
Such CP maps are called purely generated \cite{Fannes1992,FANNES1994511}. 

\begin{proposition}(Proposition 3.5 of \cite{Matsui2001})\label{Prop:Matsui}
  Let $\omega$ be a translation invariant pure split state on $\CA_{\IZ}$. Then, there exists an isometry $V$ as in \eqref{Popescu isometry} as well as a normal state $\phi$ in $\B\CV$ such that $(\EE_V,\rho)$ is an SFCS representing $\omega$.
  
 We also have $\lim_{N\to \infty}\,\Phi^N(x)=\phi(x)\cdot \1$ for any $x\in \B\CV$ and the vN algebra generated by the Kraus operators is $\B\CV$.
\end{proposition}
This proposition, in particular, implies that there is no loss of information if we go from $(\EE,\phi)$ to $(\Phi,\phi)$: $\Phi$ and $\EE$ are related through the isometry V. Note also that in the case of split states, as opposed to gapped ground states of local Hamiltonians, there is no guarantee that $\Phi^N$ will converge to its (ultraweak) limit exponentially fast.

The generalized MPS tensor plays the role of the linear datum associated with a 1d TI state. We note that in order to obtain a well-defined state from a gMPS tensor, it must take values in the set of Hilbert-Schmidt operators $\HS\CV\subset \B\CV$.

\subsection{RG fixed points and trivial topological phases}

We wish to study topological phases through the lens of the renormalization group fixed points. Let $\omega$ be a state presented by a pair $(\Phi,\phi)$ or SFCS $(\EE,\phi)$. Renormalization of quantum states in 1d comprises of blocking neighboring sites and removing local entanglement via unitary transformations \cite{Verstraete_2005}. The SFCS formalism allows to define the RG flow in a coordinate-free way: the action on $\Phi$ extracted from $\omega$ is simply $\Phi\mapsto \Phi\circ \Phi$. Upon blocking of sites, the new local Hilbert space has the dimension equal to the number of Kraus operators representing $\Phi\circ \Phi$. Therefore, the local Hilbert space either does not grow if $\CV=\IC$, or bounded by $(\mbox{dim}\CV)^2$. As an example, the fixed point obtained from the AKLT model \cite{Fannes1992} is given by $\CV=\IC^2$, $\CP\cong \IC^4$, and $\rho=\1$. 

States that can be represented by SFCS with $\dim \CV<\infty$ are known in the literature as finitely correlated states \cite{Fannes1992}. From the point of view of infinite spin chains, this case is somewhat degenerate as it assumes that there is a finite number of non-zero Schmidt coefficients in the decomposition of such states into the left and right half-chains (the split property allows us to define the Schmidt coefficients).

In what follows we assume that $\dim\CV=\infty$. Thus, at the RG fixed points, the local Hilbert space dimension is infinite. Strictly speaking, $\Phi^N$ as $N\to \infty$ cannot be a part of the SFCS data as it corresponds to $\dim \CP=\infty$. The CP map corresponding to RG fixed point is the ultraweak limit of $\Phi^N$ as  $N\to\infty$. 

\begin{remark}
   The SFCS construction as is does not apply to quantum lattice systems with infinite dimensional local Hilbert spaces. We need to require that the normal CP map $\Phi$ in Definition \ref{def:SFCS} is \textit{adjointable} \cite{GMM}. We say that the normal CP map is \textit{adjointable} if it preserves the space of trace-class operators $\T\CV$ and there exists an nCP map $\Phi^{\dagger}$, called the adjoint of $\Phi$, on $\B\CV$ such that $\Tr[\Phi^{\dagger}(\kappa)\,x]=\Tr[ \kappa\, \Phi(x)]$. The adjointable CP maps do define generalized MPS tensors \cite{GMM}.
\end{remark}

Proposition \ref{Prop:Matsui} describes RG fixed points of all TI pure split states: the are defined by conditional expectations $F^2=F$ onto the fixed point subspace $\IC\cdot \1$. Note $F$ has an adjoint, in the sense of the Remark above, and $F^{\dagger}(\kappa)=\Tr[\kappa]\cdot \rho$ for any $\kappa \in \T\CV$. Physically, such fixed points correspond to a completely clustered state on $\CA_{\IZ}$ with zero correlation length.  Our fixed points are straightforward generalizations of the fixed points described in \cite{Verstraete_2005,Shiozaki:2016} where $\CV$ is assumed finite-dimensional.

 It is not difficult to obtain the Kraus decomposition for $F$. We denote the Hilbert space of \textit{Hilbert-Schmidt} operators on $\CV$ by $\HS\CV$. Let us choose an ON basis $\{e_i\}_{i=1}^{\dim \CV}$ for $\CV$ such that $\{e_i\otimes e_j^{\vee}\}_{i,j=1}^{\dim \CV}$ is a basis for $\HS\CV$. At the RG fixed point, the physical Hilbert space $\CP\cong \HS\CV$. The generalized MPS tensor is constructed by any isomorphism $t:\CP \xrightarrow{\sim} \HS\CV$:  
\begin{align}\label{RGfixedMPS}
    T\circ t^{-1}(x)= x\,\rho^{1/2}\,,\quad x\in \HS\CV\,.
\end{align}

These fixed point tensors are equivalent in the sense that their images generate the same vN algebra $\B\CV$ or the same \Cs-algebra of compact operators $\K\CV$. For our purposes, we are going to use the \Cs-completion. Note that $\K\CV$ is Morita equivalent to $\IC$ in the sense that there exists an imprimitivity $\CK\CV-\IC$ bimodule \cite{Raeburn1998MoritaEA}. Therefore, all fixed points of pure split states are equivalent as their generalized MPS tensors generate either a matrix algebra or $\K\CV$, which are all Morita-equivalent to $\IC$.

\subsection{Families of RG fixed states}\label{sec;familiesofRG}  We consider families of topologically trivial physical states continuously parametrized by a topological space $M$. For simplicity, we assume $\CM$ to be a compact manifold and each state in the family to be an RG fixed state. 

A continuous family of pure quantum mechanical states can be parametrized by sections of a not necessarily trivial hermitian line bundle \cite{Simon1983}--this is the usual setup for  ordinary Berry curvature \cite{Berry}. The non-vanishing Berry curvature is an obstruction to the existence of a continuous vector field in the Hilbert space presenting the pure state. Similarly, higher Berry curvature is an obstruction for a family of one-dimensional states to admit a continuous representation by a field of MPS tensors, assuming each element of the family admits such a representation \cite{ohyama2024higher,qi2023charting}. 

Due to the theorem by Hastings \cite{hastings2007area}, any unique ground state of a local gapped Hamiltonian can be \textit{approximated} by an MPS with a finite-dimensional bond space. In our approach, we use generalized MPS exactly presenting any split pure state. Families of injective MPS with bond space $\CV$ of constant dimension along the parameter space correspond to principle bundles with the structure group $\operatorname{PGL}(\CV)$ \cite{ohyama2024discrete,ohyama2024higher,qi2023charting}. Such bundles can realize only the torsion part of $H^3(\CM,\IZ)$, due to the Serre theorem \cite{Serre}. In order to realize non-torsion classes, \cite{ohyama2024higher,qi2023charting} consider families of MPS with bond spaces dimensions varying from patch to patch. We employ generalized MPS with a constant but infinite bond dimension which allow us to realize the non-torsion classes. 

A generalized RG fixed MPS tensor is defined up to an overall conjugation with a unitary. Let $T,T':\CP\to \B\CV$ be gMPS tensors such that
\begin{align}
 T'(\psi)=U\,T(\psi)\,U^*\,,\quad U\in \operatorname{U}(\CV)\,.
\end{align}
The corresponding pairs $(\Phi,\phi)$ and $(\operatorname{Ad}_U\circ\Phi, \phi \circ \operatorname{Ad}_{U^*})$ define the same state on $\CA_{\IZ}$. This gauge redundancy allows us to glue families of gMPS tensors from the local data. 

For simplicity, we consider only RG fixed tensors. In order to define a family of RG fixed gMPS tensors, we first define locally-trivial bundles of Hilbert-Schmidt operators. Let $\CH$ be the infinite dimensional separable Hilbert space -- this space will play the role of $\CV$ from the previous section. A locally-trivial bundle of Hilbert-Schmidt operators on $\CH$ (HS bundle) $E$ trivialized over a good cover $\CU$ is a bundle with a typical fiber $\HS\CH$:
\begin{align}
    E|_{\CU_{\a}}\cong \CU_{\a}\times \HS \CH\,.
\end{align}
The structure group of $E$ is the projective unitary group $\bP\operatorname{U}(\CH)$ as $\mbox{Aut}(\HS \CH)=\bP\operatorname{U}(\CH)$. There are many non-equivalent topologies on $\operatorname{U}(\CH)$ and $\bP\operatorname{U}(\CH)$: we choose the compact-open topology, see details in \cite{atiyah2004twisted}. In this topology, $\operatorname{U}(\CH)$ acts continuously on the algebra $\HS \CH$ by conjugations (Prop. A1.1 of \cite{atiyah2004twisted}) and is contractible (Prop. A2.1, loc. cit.). 

\begin{proposition}
    The set of isomorphism classes of bundles of Hilbert-Schmidt operators on $\CH$ is isomorphic to the Abelian groups 
    \begin{align}
        H^1(\CM,\underline{\bP\operatorname{U}(\CH)})\cong  H^2(\CM,\underline{\operatorname{U}(1)})\cong H^3(\CM,\IZ)\,.
    \end{align}
\end{proposition}
\begin{proof}
The proof of the analogous classification result for bundles of algebras of compact operators is given in \cite{DD} and reviewed in \cite{Brylinski}. The proof of the present proposition repeats verbatim after we notice that $\mbox{Aut}(\HS\CH)=\bP\operatorname{U}(\CH)$.
\end{proof}

According to \eqref{RGfixedMPS}, the parametrized RG-fixed gMPS tensor over $\CU_{\a}$ is defined by a local frame for a HS bundle. In other words, there exists a global gMPS parametrization of a family of RG fixed states if and only if there exists a global frame for the HS bundle. The latter exists if and only if the HS bundle is trivial, i.e., the corresponding cohomology class in $H^3(\CM,\IZ)$ vanishes. 

Finally, we want to note that a similar analysis can be conducted for the general injective, but not necessarily RG fixed, gMPS tensors. 
\subsection{Projective Hilbert bundles}
\begin{comment}
By the standard argument, the short exact sequence 
\begin{align}
\operatorname{U}(1)\longrightarrow \operatorname{U}(\CH) \longrightarrow \bP\operatorname{U}(\CH)
\end{align}
implies that $\bP\operatorname{U}(\CH)$ has the homotopy type of $K(\IZ,2)$ and the classifying space $B\bP\operatorname{U}(\CH)$ has the homotopy type of $K(\IZ,3)$. Therefore, principle $B\bP\operatorname{U}(\CH)$-bundles over $\CM$ are classified by $H^3(\CM,\IZ)$. 
\end{comment}

Let us describe how all classes of HS bundles can be constructed. This construction will be specifically useful for the RG fixed points.

 Consider a locally trivial bundle $P$ with a fiber $\bP\CH$, where $\mathbb{P}\CH$ is a projectivization of $\CH$, and the structure group $\bP \operatorname{U}(\CH)$ over $\CM$.  Locally, the total space of this bundle is isomorphic to $\CU_{\a}\times \bP\CH$. One can ask\footnote{This problem was posed and solved by G. Segal in an unpublished work.} if there exists a bundle $W$ of Hilbert spaces  such that $P=\bP W$. The obstruction to the lifting of $P$ is a twisting class $\eta_{P}\in H^2(\CM,\underline{\operatorname{U}(1)})\cong H^3(\CM,\IZ)$, see, e.g., Proposition 2.1 of \cite{atiyah2004twisted}.
 
 Any bundle of projective Hilbert spaces $P$ has a dual such that $\eta_{P^*}=-\eta_{P}$. Given a bundle of projective Hilbert spaces $P$, we can form a bundle $P\otimes P^*$ with vanishing twisting class. As the lifting obstruction of $P\otimes P^*$ vanishes, it can be lifted to a bundle of Hilbert spaces which is nothing but the bundle of Hilbert-Schmidt operators $P\otimes P^{*} \cong \bP (\HS P) $. It is easy to check that the lifting obstruction $\eta_P \in H^3(\CM,\IZ)$ of $P$ is the same as the isomorphism class of the corresponding HS bundle. This way, any projective Hilbert bundle gives rise to a bundle of HS operators. This correspondence should be seen as a version of the \textit{bulk-boundary} correspondence. Indeed, in finite volume, boundary conditions for a family of gMPS are given by a (projective) Hilbert bundle, a module over the HS bundle. There exists a family of Hilbert spaces hosting the boundary conditions if and only if the twisting class of a bundle of projective spaces vanishes.

Let remark on the similarity of this construction with the standard construction of 1d SPTs: given a projective representation $v$ of $G$, we define the local Hilbert spaces as $\CP=v\otimes v^*$ and entangle them by pairing neighboring $v^*$ and $v$, as was done in \cite{chen2013symmetry,Shiozaki:2016}. 

\section{Topological T-duality}
The T-duality is one of the first dualities discovered in string theory where it relates  different toroidal compactifications, see, e.g., \cite{Becker_Becker_Schwarz_2006} for a textbook account. The T-duality has a remnant in (1+1)-dimensional gauged non-linear sigma models \cite{BUSCHER1988466} whose target-space is a circle bundle. The circle action on the target space of an NLSM is a global symmetry which can be gauged; the T-duality of NLSMs is an exact invertible transformation of the path integral which corresponds to gauging the circle action. The NLSM consists of three pieces of data $(\CM,B,g)$ where $\CM$ is a smooth manifold, which we assume to be a total space of a circle bundle, $B$ is a B-field defining a class in $H^3(\CM,\IZ)$, and $g$ is a metric. Then, the T-duality then is a rule, aka the Buscher rules \cite{BUSCHER1988466}, associating another triple  $(\CM,B,g)\leftrightarrow (\widehat \CM,\widehat B, \widehat g)$.
 This is called the \textbf{geometric T-duality} as it is formulated with the use of a metric. Crucially, the topology of $\CM$ can be different from the topology of $\widehat \CM$ \cite{Bouwknegt_2004}. The rule of topology change is described below.
 
The T-duality can be formulated purely topologically, without the use of a metric; this version is called \textbf{topological T-duality}. The (Abelian) topological T-duality after Bunke and Schick is formulated in terms of a pair $(\CM, H)$ where $\CM$ is a total space of a circle bundle over the base $\CB$, and a cohomology class $H\in H^{3}(\CM,\IZ)$. The T-duality prescribes that there exists another pair $(\widehat \CM,\widehat H)$ where $\widehat \CM$ is a total space of another circle bundle over the same base $\CB$ and a class $\widehat H\in H^{3}(\widehat \CM,\IZ)$ (Lemma 2.32 of \cite{bunke2005topology}).
 
 As a circle bundle, $\CM$ is completely determined by the first Chern class $c_1(\CM)\in \mbox{H}^2(\CB,\IZ)$ of the associated complex line bundle. The T-dual space $\WCM$ is another circle bundle and the Chern classes of the T-dual spaces are related via \cite{Bouwknegt_2004,bunke2005topology}:
 \begin{align}\label{T-dualityRule}
c_1(\CM)=\hat\pi_* \widehat{H}\,,\quad c_1(\widehat{\CM})=\pi_*H
 \end{align}
 where $\pi_*$ and $ \hat \pi_*$ are the push-forward maps (integrations along the fibers).
 Further, the T-dual pairs must fit into the following diagram \cite{Bouwknegt_2004,BUNKE_2006}:
  \begin{center}
\begin{tikzpicture}[scale=1.3]
\node (A) at (-1,1) {$p^*H$};
\node (B) at (1,1) {$\widehat{p}^{\,*}\widehat{H}$};
\node (C) at (-2,0) {$H$};
\node (D) at (2,0) {$\widehat{H}$};
\node (E) at (0,0) {$\CM\times_{\CB}\widehat{\mathcal{M}}$};
\node (F) at (-1,-1) {$\CM$};
\node (G) at (1,-1) {$\widehat{\mathcal{M}}$};
\node (H) at (0,-2) {$\CB$};
%\path[->,font=\scriptsize,>=angle 90]
\path[commutative diagrams/.cd, every arrow, every label]
(A) edge node[above]{$=$} (B)
(A) edge (C)
(B) edge (D)
(E) edge node[right]{$\,\,\widehat{p}$} (G)
(E) edge node[left]{$p$} (F)
(F) edge node[left]{$\pi$} (H)
(G) edge node[right]{$\,\,\widehat{\pi}$}(H);
\end{tikzpicture}
 \end{center}
Here $p$ and $\widehat{p}$ are the projections onto $\CM$ and $\widehat{\CM}$ respectively and $p^{\,*}H$, $\widehat{p}^{\,*}\widehat{H}$ are the corresponding pull-backs of $H$ and $\widehat{H}$. The horizontal arrow is an equality of the cohomology classes. 

\subsection{Examples}   
Let us consider two examples with $\CM=S^3$. Let $\omega^{S^1}$, $\omega^{S^2}$, and $\omega^{S^3}$ be the volume forms on $S^1$, $S^2$, and $S^3$ respectively. Consider the following pair $(S^3, \omega^{S^3}$). The $S^3$ is the Hopf bundle $\pi:S^3\to S^2$ with $c_1(S^3)=\omega^{S^2}$. It is a known fact that $\pi_*\omega^{S^3}=\omega^{S^2}$. As $c_1(\widehat \CM)\cong \omega^{S^2}$, we imply that $\widehat{\CM}\cong S^3$ and that $\widehat H=\omega ^{S^3}$. Thus, we obtain the following T-dual pair 
\begin{align}
    (S^3, \omega^{S^3})\leftrightarrow (S^3, \omega^{S^3})\,.
\end{align}

For the second example, we take $(S^3,H=0)$. Then, $\widehat \CM=S^2\times S^1$. The condition $\widehat{\pi}_*\widehat{H}=\omega^{S^2}$ implies that $\widehat H=\omega^{S^2}\otimes \omega^{S^1}$ where $\omega^{S^2}\otimes \omega^{S^1}$ is the generator of $\mbox{H}^3(S^2\times S^1,\IZ)$. We obtain the second pair
\begin{align}
    (S^3, 0)\leftrightarrow (S^2\times S^1, \omega^{S^2}\otimes \omega^{S^1})\,.
\end{align}
As a generalization of the previous examples, consider the lens spaces $L(n;1)$. One obtains the following pairs \cite{Bouwknegt_2004}:
\begin{align}
    (L(n;1), m\,\omega)\leftrightarrow  (L(m;1),n\,\omega')\,
\end{align}
where $\omega$ and $\omega'$ are the generators of $H^3(L(n;1),\IZ)$ and $H^3(L(m;1),\IZ)$ respectively.

\subsection{T-duality via continuous-trace algebras} 
Informally, the T-duality is a gauging of a circle action on a circle bundle. In this section we review the operator algebra approach to T-duality so that it can be seen as a transformation of an algebra generated by parametrized gMPS tensors.

We are interested in \Cs-algebras generated by the gMPS tensors. In order to access the tools available for the Banach algebras, we complete the HS bundles in the operator norm topology. The completion of $\HS\CH$ is the \Cs-algebra of compact operators $\K\CH$ such that the  bundle of Hilbert-Schmidt operators upon completion becomes a bundle with the typical fiber $\K\CH$ and structure group $\bP\operatorname{U}(\CH)$. A theorem by Dixmier and Douady \cite{DD} states that the \Cs-algebra of sections of a bundle of compact operators over $\CM$ (assumed to be a compact manifold) is isomorphic to a stable continuous-trace (CT) algebra with spectrum $\CM$. A CT algebra $A$ is stable if $A\otimes \K\CH\cong A$. The CT algebras can be defined internally but we use their connection with bundles of compact operators as a definition, see also \cite{Rosenberg_1989}. 

\begin{remark}
    The Gelfand duality \cite{GelfandNeumark1943,NEGREPONTIS1971228} establishes a correspondence between  commutative  \Cs  algebras and locally compact Hausdorff topological spaces, such that any such algebra $A$ is isomorphic to the algebra of continuous complex functions on a compact topological space $A\cong C_0(X)$, where $X$ is the Gelfand spectrum of $A$. Informally, the CT algebras locally look like $C_0(\CM)\otimes \K\CH$ \cite{Raeburn1998MoritaEA}.
\end{remark} 

The CT algebras with the spectrum $\CM$ are classified by the Dixmier-Doudy invariant \cite{DD} taking values in $H^3(\CM,\IZ)$. Let us denote by $CT(\CM,H)$ the isomorphism class of CT algebras with spectrum $\CM$ and Dixmier-Douady invariant $H\in H^3(\CM,\IZ)$.

A crucial for our applications is a theorem by Philips and Raeburn \cite{RaeburnRosenberg1988} that says that the result of a gauging of a free $\operatorname{U}(1)$ action on a CT algebra is another CT algebra. Practically, we are interested in compact manifolds, which are automatically second-countable (see, e.g. \cite{LeeTop2010}) and have homotopy types of finite CW complexes \cite{Kirby1969OnTT}. 
\begin{theorem}(Raeburn-Rosenberg \cite{RaeburnRosenberg1988,Raeburn1998MoritaEA,rosenbergtopology})
    Let $\CM$ be a compact topological space which is second countable and has the homotopy type of a finite CW complex. Let $\CM$ be a total space of a circle bundle $\pi:\CM\to\CB$ and $A=CT(\CM,H)$ with $H\in H^3(\CM,\IZ)$. Lift the free action of $S^1\cong \IR/\IZ$ on $\CM$ to a locally free action of $\IR$ on $\CM$, and then to an action $\alpha$ of $\IR$ on $A$. Then, the $B=A\rtimes_{\alpha}\IR$ is a stable $CT$-algebra with spectrum $\widehat{\CM}$, the total space of another  circle bundle $\widehat{\pi}:\widehat{\CM}\to \CB$. The Chern classes and the DD classes are related as in \eqref{T-dualityRule}.
\end{theorem}
The auxiliary results proving the existence of the lifts can be found in \cite{rosenbergtopology}. 

In the operator-algebraic approach, the T-duality is a transformation of the algebra of sections of bundles of compacts operators. We obtained the latter as a \Cs-completion of the bundles of RG fixed gMPS tensors. The gauging of the $\operatorname{U}(1)$ action indeed reproduces the topological T-duality reviewed in the previous section.

\hfill

\bmhead{Acknowledgements} The author is indebted to Yichen Hu who proposed to study the T-duality of families of gapped systems.  
The author is grateful to Tom Mainiero for many valuable lessons in operator algebras and for reading the draft, and to Greg Moore, Shuhei Ohyama, Shinsei Ryu, and Ryan Thorngren for fruitful discussions. The project began when the author was supported by the DOE under grant DOE-
SC0010008 to Rutgers. When working on this project the author visited the Aspen Center for Physics, which is supported by National Science Foundation grant PHY-2210452. The author acknowledges support by the Mani L. Bhaumik Institute for Theoretical Physics and by the DOE under grant DOE-
SC0010008 to Rutgers.

\backmatter

\bibliography{sn-bibliography}


\begin{thebibliography}{62}
\ifx \bisbn   \undefined \def \bisbn  #1{ISBN #1}\fi
\ifx \binits  \undefined \def \binits#1{#1}\fi
\ifx \bauthor  \undefined \def \bauthor#1{#1}\fi
\ifx \batitle  \undefined \def \batitle#1{#1}\fi
\ifx \bjtitle  \undefined \def \bjtitle#1{#1}\fi
\ifx \bvolume  \undefined \def \bvolume#1{\textbf{#1}}\fi
\ifx \byear  \undefined \def \byear#1{#1}\fi
\ifx \bissue  \undefined \def \bissue#1{#1}\fi
\ifx \bfpage  \undefined \def \bfpage#1{#1}\fi
\ifx \blpage  \undefined \def \blpage #1{#1}\fi
\ifx \burl  \undefined \def \burl#1{\textsf{#1}}\fi
\ifx \doiurl  \undefined \def \doiurl#1{\url{https://doi.org/#1}}\fi
\ifx \betal  \undefined \def \betal{\textit{et al.}}\fi
\ifx \binstitute  \undefined \def \binstitute#1{#1}\fi
\ifx \binstitutionaled  \undefined \def \binstitutionaled#1{#1}\fi
\ifx \bctitle  \undefined \def \bctitle#1{#1}\fi
\ifx \beditor  \undefined \def \beditor#1{#1}\fi
\ifx \bpublisher  \undefined \def \bpublisher#1{#1}\fi
\ifx \bbtitle  \undefined \def \bbtitle#1{#1}\fi
\ifx \bedition  \undefined \def \bedition#1{#1}\fi
\ifx \bseriesno  \undefined \def \bseriesno#1{#1}\fi
\ifx \blocation  \undefined \def \blocation#1{#1}\fi
\ifx \bsertitle  \undefined \def \bsertitle#1{#1}\fi
\ifx \bsnm \undefined \def \bsnm#1{#1}\fi
\ifx \bsuffix \undefined \def \bsuffix#1{#1}\fi
\ifx \bparticle \undefined \def \bparticle#1{#1}\fi
\ifx \barticle \undefined \def \barticle#1{#1}\fi
\bibcommenthead
\ifx \bconfdate \undefined \def \bconfdate #1{#1}\fi
\ifx \botherref \undefined \def \botherref #1{#1}\fi
\ifx \url \undefined \def \url#1{\textsf{#1}}\fi
\ifx \bchapter \undefined \def \bchapter#1{#1}\fi
\ifx \bbook \undefined \def \bbook#1{#1}\fi
\ifx \bcomment \undefined \def \bcomment#1{#1}\fi
\ifx \oauthor \undefined \def \oauthor#1{#1}\fi
\ifx \citeauthoryear \undefined \def \citeauthoryear#1{#1}\fi
\ifx \endbibitem  \undefined \def \endbibitem {}\fi
\ifx \bconflocation  \undefined \def \bconflocation#1{#1}\fi
\ifx \arxivurl  \undefined \def \arxivurl#1{\textsf{#1}}\fi
\csname PreBibitemsHook\endcsname

%%% 1
\bibitem[\protect\citeauthoryear{Accardi}{1974}]{accardi1974}
\begin{botherref}
\oauthor{\bsnm{Accardi}, \binits{L.}}:
Non-commutative {M}arkov chains.
Proceedings International School of Mathematical Physics - Universita' di Camerino 30 Sept. - 12 Oct. (1974)
(1974)
\end{botherref}
\endbibitem

%%% 2
\bibitem[\protect\citeauthoryear{Bratteli et~al.}{1996}]{bratteli1996endomorphisms}
\begin{bchapter}
\bauthor{\bsnm{Bratteli}, \binits{O.}},
\bauthor{\bsnm{Jorgensen}, \binits{P.E.}},
\bauthor{\bsnm{Price}, \binits{G.L.}}:
\bctitle{Endomorphisms of {B(H)}}.
In: \bbtitle{Proceedings of Symposia in Pure Mathematics},
vol. \bseriesno{59},
pp. \bfpage{93}--\blpage{138}
(\byear{1996}).
\bcomment{American Mathematical Society}
\end{bchapter}
\endbibitem

%%% 3
\bibitem[\protect\citeauthoryear{Bratteli et~al.}{2000}]{bratteli2000pure}
\begin{botherref}
\oauthor{\bsnm{Bratteli}, \binits{O.}},
\oauthor{\bsnm{Jorgensen}, \binits{P.E.}},
\oauthor{\bsnm{Kishimoto}, \binits{A.}},
\oauthor{\bsnm{Werner}, \binits{R.F.}}:
Pure states on {O}$_d$.
Journal of operator theory,
97--143
(2000)
\end{botherref}
\endbibitem

%%% 4
\bibitem[\protect\citeauthoryear{Matsui}{2001}]{Matsui2001}
\begin{barticle}
\bauthor{\bsnm{Matsui}, \binits{T.}}:
\batitle{The split property and the symmetry breaking of the quantum spin chain}.
\bjtitle{Commun. Math. Phys.}
\bvolume{218},
\bfpage{393}--\blpage{416}
(\byear{2001})
\end{barticle}
\endbibitem

%%% 5
\bibitem[\protect\citeauthoryear{Fannes et~al.}{1992}]{Fannes1992}
\begin{barticle}
\bauthor{\bsnm{Fannes}, \binits{M.}},
\bauthor{\bsnm{Nachtergaele}, \binits{B.}},
\bauthor{\bsnm{Werner}, \binits{R.F.}}:
\batitle{{Finitely correlated states on quantum spin chains}}.
\bjtitle{Communications in Mathematical Physics}
\bvolume{144}(\bissue{3}),
\bfpage{443}--\blpage{490}
(\byear{1992})
\doiurl{cmp/1104249404}
\end{barticle}
\endbibitem

%%% 6
\bibitem[\protect\citeauthoryear{Perez-Garcia et~al.}{2006}]{perez2006matrix}
\begin{botherref}
\oauthor{\bsnm{Perez-Garcia}, \binits{D.}},
\oauthor{\bsnm{Verstraete}, \binits{F.}},
\oauthor{\bsnm{Wolf}, \binits{M.M.}},
\oauthor{\bsnm{Cirac}, \binits{J.I.}}:
Matrix product state representations
(2006)
\end{botherref}
\endbibitem

%%% 7
\bibitem[\protect\citeauthoryear{Cirac et~al.}{2021}]{cirac2021matrix}
\begin{barticle}
\bauthor{\bsnm{Cirac}, \binits{J.I.}},
\bauthor{\bsnm{Perez-Garcia}, \binits{D.}},
\bauthor{\bsnm{Schuch}, \binits{N.}},
\bauthor{\bsnm{Verstraete}, \binits{F.}}:
\batitle{Matrix product states and projected entangled pair states: Concepts, symmetries, theorems}.
\bjtitle{Reviews of Modern Physics}
\bvolume{93}(\bissue{4}),
\bfpage{045003}
(\byear{2021})
\end{barticle}
\endbibitem

%%% 8
\bibitem[\protect\citeauthoryear{Sathiapalan}{1987}]{PhysRevLett.58.1597}
\begin{barticle}
\bauthor{\bsnm{Sathiapalan}, \binits{B.}}:
\batitle{Duality in statistical mechanics and string theory}.
\bjtitle{Phys. Rev. Lett.}
\bvolume{58},
\bfpage{1597}--\blpage{1599}
(\byear{1987})
\doiurl{10.1103/PhysRevLett.58.1597}
\end{barticle}
\endbibitem

%%% 9
\bibitem[\protect\citeauthoryear{Buscher}{1988}]{BUSCHER1988466}
\begin{barticle}
\bauthor{\bsnm{Buscher}, \binits{T.H.}}:
\batitle{Path-integral derivation of quantum duality in nonlinear sigma-models}.
\bjtitle{Physics Letters B}
\bvolume{201}(\bissue{4}),
\bfpage{466}--\blpage{472}
(\byear{1988})
\doiurl{10.1016/0370-2693(88)90602-8}
\end{barticle}
\endbibitem

%%% 10
\bibitem[\protect\citeauthoryear{Abanov and Wiegmann}{2000}]{abanov2000theta}
\begin{barticle}
\bauthor{\bsnm{Abanov}, \binits{A.}},
\bauthor{\bsnm{Wiegmann}, \binits{P.B.}}:
\batitle{Theta-terms in nonlinear sigma-models}.
\bjtitle{Nuclear Physics B}
\bvolume{570}(\bissue{3}),
\bfpage{685}--\blpage{698}
(\byear{2000})
\end{barticle}
\endbibitem

%%% 11
\bibitem[\protect\citeauthoryear{Cordova et~al.}{2020}]{cordova2020anomalies}
\begin{barticle}
\bauthor{\bsnm{Cordova}, \binits{C.}},
\bauthor{\bsnm{Freed}, \binits{D.}},
\bauthor{\bsnm{Lam}, \binits{H.T.}},
\bauthor{\bsnm{Seiberg}, \binits{N.}}:
\batitle{Anomalies in the space of coupling constants and their dynamical applications i}.
\bjtitle{SciPost Physics}
\bvolume{8}(\bissue{1}),
\bfpage{001}
(\byear{2020})
\end{barticle}
\endbibitem

%%% 12
\bibitem[\protect\citeauthoryear{Hsin et~al.}{2020}]{hsin2020berry}
\begin{barticle}
\bauthor{\bsnm{Hsin}, \binits{P.-S.}},
\bauthor{\bsnm{Kapustin}, \binits{A.}},
\bauthor{\bsnm{Thorngren}, \binits{R.}}:
\batitle{Berry phase in quantum field theory: {D}iabolical points and boundary phenomena}.
\bjtitle{Physical Review B}
\bvolume{102}(\bissue{24}),
\bfpage{245113}
(\byear{2020})
\end{barticle}
\endbibitem

%%% 13
\bibitem[\protect\citeauthoryear{Bouwknegt et~al.}{2004}]{Bouwknegt_2004}
\begin{barticle}
\bauthor{\bsnm{Bouwknegt}, \binits{P.}},
\bauthor{\bsnm{Evslin}, \binits{J.}},
\bauthor{\bsnm{Mathai}, \binits{V.}}:
\batitle{T-duality: Topology change from {H}-flux}.
\bjtitle{Communications in Mathematical Physics}
\bvolume{249}(\bissue{2}),
\bfpage{383}--\blpage{415}
(\byear{2004})
\doiurl{10.1007/s00220-004-1115-6}
\end{barticle}
\endbibitem

%%% 14
\bibitem[\protect\citeauthoryear{Bunke and Schick}{2005}]{bunke2005topology}
\begin{barticle}
\bauthor{\bsnm{Bunke}, \binits{U.}},
\bauthor{\bsnm{Schick}, \binits{T.}}:
\batitle{On the topology of {T}-duality}.
\bjtitle{Reviews in Mathematical Physics}
\bvolume{17}(\bissue{01}),
\bfpage{77}--\blpage{112}
(\byear{2005})
\end{barticle}
\endbibitem

%%% 15
\bibitem[\protect\citeauthoryear{Moore and Segal}{2006}]{moore2006dbranes}
\begin{botherref}
\oauthor{\bsnm{Moore}, \binits{G.W.}},
\oauthor{\bsnm{Segal}, \binits{G.}}:
D-branes and {K}-theory in 2D topological field theory
(2006)
\end{botherref}
\endbibitem

%%% 16
\bibitem[\protect\citeauthoryear{Tachikawa}{2020}]{tachikawa2020gauging}
\begin{barticle}
\bauthor{\bsnm{Tachikawa}, \binits{Y.}}:
\batitle{On gauging finite subgroups}.
\bjtitle{SciPost Physics}
\bvolume{8}(\bissue{1}),
\bfpage{015}
(\byear{2020})
\end{barticle}
\endbibitem

%%% 17
\bibitem[\protect\citeauthoryear{Gaiotto and Kulp}{2021}]{Gaiotto_2021}
\begin{botherref}
\oauthor{\bsnm{Gaiotto}, \binits{D.}},
\oauthor{\bsnm{Kulp}, \binits{J.}}:
Orbifold groupoids.
Journal of High Energy Physics
\textbf{2021}(2)
(2021)
\doiurl{10.1007/jhep02(2021)132}
\end{botherref}
\endbibitem

%%% 18
\bibitem[\protect\citeauthoryear{Fuchs et~al.}{2008}]{FUCHS2008576}
\begin{barticle}
\bauthor{\bsnm{Fuchs}, \binits{J.}},
\bauthor{\bsnm{Schweigert}, \binits{C.}},
\bauthor{\bsnm{Waldorf}, \binits{K.}}:
\batitle{Bi-branes: Target space geometry for world sheet topological defects}.
\bjtitle{Journal of Geometry and Physics}
\bvolume{58}(\bissue{5}),
\bfpage{576}--\blpage{598}
(\byear{2008})
\doiurl{10.1016/j.geomphys.2007.12.009}
\end{barticle}
\endbibitem

%%% 19
\bibitem[\protect\citeauthoryear{Dove and Schick}{2024}]{Dove:2023pqy}
\begin{barticle}
\bauthor{\bsnm{Dove}, \binits{T.}},
\bauthor{\bsnm{Schick}, \binits{T.}}:
\batitle{{Equivariant Topological T-Duality}}.
\bjtitle{Commun. Math. Phys.}
\bvolume{405}(\bissue{8}),
\bfpage{179}
(\byear{2024})
\doiurl{10.1007/s00220-024-05044-0}
{\href{https://arxiv.org/abs/2310.06064}{{arXiv:2310.06064}}}
\end{barticle}
\endbibitem

%%% 20
\bibitem[\protect\citeauthoryear{Daenzer}{2007}]{daenzer2007groupoidapproachnoncommutativetduality}
\begin{botherref}
\oauthor{\bsnm{Daenzer}, \binits{C.}}:
A groupoid approach to noncommutative T-duality
(2007).
\url{https://arxiv.org/abs/0704.2592}
\end{botherref}
\endbibitem

%%% 21
\bibitem[\protect\citeauthoryear{Pollmann et~al.}{2010}]{Pollmann2010}
\begin{barticle}
\bauthor{\bsnm{Pollmann}, \binits{F.}},
\bauthor{\bsnm{Turner}, \binits{A.M.}},
\bauthor{\bsnm{Berg}, \binits{E.}},
\bauthor{\bsnm{Oshikawa}, \binits{M.}}:
\batitle{Entanglement spectrum of a topological phase in one dimension}.
\bjtitle{Phys. Rev. B}
\bvolume{81},
\bfpage{064439}
(\byear{2010})
\doiurl{10.1103/PhysRevB.81.064439}
\end{barticle}
\endbibitem

%%% 22
\bibitem[\protect\citeauthoryear{Chen et~al.}{2013}]{chen2013symmetry}
\begin{barticle}
\bauthor{\bsnm{Chen}, \binits{X.}},
\bauthor{\bsnm{Gu}, \binits{Z.-C.}},
\bauthor{\bsnm{Liu}, \binits{Z.-X.}},
\bauthor{\bsnm{Wen}, \binits{X.-G.}}:
\batitle{Symmetry protected topological orders and the group cohomology of their symmetry group}.
\bjtitle{Physical Review B}
\bvolume{87}(\bissue{15}),
\bfpage{155114}
(\byear{2013})
\end{barticle}
\endbibitem

%%% 23
\bibitem[\protect\citeauthoryear{Ogata}{2021}]{ogata2019classification}
\begin{barticle}
\bauthor{\bsnm{Ogata}, \binits{Y.}}:
\batitle{A classification of pure states on quantum spin chains satisfying the split property with on-site finite group symmetries}.
\bjtitle{Transactions of the American Mathematical Society, Series B}
\bvolume{8}(\bissue{2}),
\bfpage{39}--\blpage{65}
(\byear{2021})
\end{barticle}
\endbibitem

%%% 24
\bibitem[\protect\citeauthoryear{Kapustin et~al.}{2021}]{kapustin2021classification}
\begin{botherref}
\oauthor{\bsnm{Kapustin}, \binits{A.}},
\oauthor{\bsnm{Sopenko}, \binits{N.}},
\oauthor{\bsnm{Yang}, \binits{B.}}:
A classification of invertible phases of bosonic quantum lattice systems in one dimension.
Journal of Mathematical Physics
\textbf{62}(8)
(2021)
\end{botherref}
\endbibitem

%%% 25
\bibitem[\protect\citeauthoryear{Ogata}{2021}]{ogata2021classification}
\begin{botherref}
\oauthor{\bsnm{Ogata}, \binits{Y.}}:
Classification of symmetry protected topological phases in quantum spin chains
(2021)
\end{botherref}
\endbibitem

%%% 26
\bibitem[\protect\citeauthoryear{Shiozaki and Ryu}{2017}]{Shiozaki:2016}
\begin{barticle}
\bauthor{\bsnm{Shiozaki}, \binits{K.}},
\bauthor{\bsnm{Ryu}, \binits{S.}}:
\batitle{{Matrix product states and equivariant topological field theories for bosonic symmetry-protected topological phases in (1+1) dimensions}}.
\bjtitle{JHEP}
\bvolume{04},
\bfpage{100}
(\byear{2017})
\doiurl{10.1007/JHEP04(2017)100}
\end{barticle}
\endbibitem

%%% 27
\bibitem[\protect\citeauthoryear{Kapustin et~al.}{2017}]{Kapustin_2017}
\begin{botherref}
\oauthor{\bsnm{Kapustin}, \binits{A.}},
\oauthor{\bsnm{Turzillo}, \binits{A.}},
\oauthor{\bsnm{You}, \binits{M.}}:
Topological field theory and matrix product states.
Physical Review B
\textbf{96}(7)
(2017)
\doiurl{10.1103/physrevb.96.075125}
\end{botherref}
\endbibitem

%%% 28
\bibitem[\protect\citeauthoryear{Hsin et~al.}{2020}]{BerryPhase2020}
\begin{barticle}
\bauthor{\bsnm{Hsin}, \binits{P.-S.}},
\bauthor{\bsnm{Kapustin}, \binits{A.}},
\bauthor{\bsnm{Thorngren}, \binits{R.}}:
\batitle{Berry phase in quantum field theory: {D}iabolical points and boundary phenomena}.
\bjtitle{Phys. Rev. B}
\bvolume{102},
\bfpage{245113}
(\byear{2020})
\doiurl{10.1103/PhysRevB.102.245113}
\end{barticle}
\endbibitem

%%% 29
\bibitem[\protect\citeauthoryear{Kapustin and Sopenko}{2022}]{kapustin2022local}
\begin{barticle}
\bauthor{\bsnm{Kapustin}, \binits{A.}},
\bauthor{\bsnm{Sopenko}, \binits{N.}}:
\batitle{Local {N}oether theorem for quantum lattice systems and topological invariants of gapped states}.
\bjtitle{Journal of Mathematical Physics}
\bvolume{63}(\bissue{9}),
\bfpage{091903}
(\byear{2022})
\end{barticle}
\endbibitem

%%% 30
\bibitem[\protect\citeauthoryear{Artymowicz et~al.}{2024}]{artymowicz2024quantization}
\begin{barticle}
\bauthor{\bsnm{Artymowicz}, \binits{A.}},
\bauthor{\bsnm{Kapustin}, \binits{A.}},
\bauthor{\bsnm{Sopenko}, \binits{N.}}:
\batitle{Quantization of the higher {B}erry curvature and the higher {T}houless pump}.
\bjtitle{Communications in Mathematical Physics}
\bvolume{405}(\bissue{8}),
\bfpage{191}
(\byear{2024})
\end{barticle}
\endbibitem

%%% 31
\bibitem[\protect\citeauthoryear{Ohyama et~al.}{2024}]{ohyama2024discrete}
\begin{barticle}
\bauthor{\bsnm{Ohyama}, \binits{S.}},
\bauthor{\bsnm{Terashima}, \binits{Y.}},
\bauthor{\bsnm{Shiozaki}, \binits{K.}}:
\batitle{Discrete higher {B}erry phases and matrix product states}.
\bjtitle{Physical Review B}
\bvolume{110}(\bissue{3}),
\bfpage{035114}
(\byear{2024})
\end{barticle}
\endbibitem

%%% 32
\bibitem[\protect\citeauthoryear{Ohyama and Ryu}{2024}]{ohyama2024higher}
\begin{barticle}
\bauthor{\bsnm{Ohyama}, \binits{S.}},
\bauthor{\bsnm{Ryu}, \binits{S.}}:
\batitle{Higher structures in matrix product states}.
\bjtitle{Physical Review B}
\bvolume{109}(\bissue{11}),
\bfpage{115152}
(\byear{2024})
\end{barticle}
\endbibitem

%%% 33
\bibitem[\protect\citeauthoryear{Qi et~al.}{2023}]{qi2023charting}
\begin{botherref}
\oauthor{\bsnm{Qi}, \binits{M.}},
\oauthor{\bsnm{Stephen}, \binits{D.T.}},
\oauthor{\bsnm{Wen}, \binits{X.}},
\oauthor{\bsnm{Spiegel}, \binits{D.}},
\oauthor{\bsnm{Pflaum}, \binits{M.J.}},
\oauthor{\bsnm{Beaudry}, \binits{A.}},
\oauthor{\bsnm{Hermele}, \binits{M.}}:
Charting the space of ground states with tensor networks.
arXiv preprint arXiv:2305.07700
(2023)
\end{botherref}
\endbibitem

%%% 34
\bibitem[\protect\citeauthoryear{Kock}{}]{kock2004frobenius}
\begin{botherref}
\oauthor{\bsnm{Kock}, \binits{J.}}:
Frobenius Algebras and 2-D Topological Quantum Field Theories.
Cambridge University Press
\end{botherref}
\endbibitem

%%% 35
\bibitem[\protect\citeauthoryear{Bott and Tu}{}]{bott2013differential}
\begin{botherref}
\oauthor{\bsnm{Bott}, \binits{R.}},
\oauthor{\bsnm{Tu}, \binits{L.W.}}:
Differential Forms in Algebraic Topology.
Graduate Texts in Mathematics.
Springer
\end{botherref}
\endbibitem

%%% 36
\bibitem[\protect\citeauthoryear{Hitchin}{2001}]{Hitchin:1999fh}
\begin{barticle}
\bauthor{\bsnm{Hitchin}, \binits{N.J.}}:
\batitle{{Lectures on special Lagrangian submanifolds}}.
\bjtitle{AMS/IP Stud. Adv. Math.}
\bvolume{23},
\bfpage{151}--\blpage{182}
(\byear{2001})
{\href{https://arxiv.org/abs/math/9907034}{{arXiv:math/9907034}}}
\end{barticle}
\endbibitem

%%% 37
\bibitem[\protect\citeauthoryear{Kristel et~al.}{2021}]{kristel20212}
\begin{botherref}
\oauthor{\bsnm{Kristel}, \binits{P.}},
\oauthor{\bsnm{Ludewig}, \binits{M.}},
\oauthor{\bsnm{Waldorf}, \binits{K.}}:
2-vector bundles
(2021)
\end{botherref}
\endbibitem

%%% 38
\bibitem[\protect\citeauthoryear{Keyl et~al.}{2006}]{keyl2006entanglement}
\begin{barticle}
\bauthor{\bsnm{Keyl}, \binits{M.}},
\bauthor{\bsnm{Matsui}, \binits{T.}},
\bauthor{\bsnm{Schlingemann}, \binits{D.}},
\bauthor{\bsnm{Werner}, \binits{R.}}:
\batitle{Entanglement, {H}aag-duality and type properties of infinite quantum spin chains}.
\bjtitle{Reviews in Mathematical Physics}
\bvolume{18}(\bissue{09}),
\bfpage{935}--\blpage{970}
(\byear{2006})
\end{barticle}
\endbibitem

%%% 39
\bibitem[\protect\citeauthoryear{Matsui}{2013}]{Matsui2013}
\begin{barticle}
\bauthor{\bsnm{Matsui}, \binits{T.}}:
\batitle{Boundedness of entanglement entropy and split property of quantum spin chains}.
\bjtitle{Reviews in Mathematical Physics}
\bvolume{25}(\bissue{09}),
\bfpage{1350017}
(\byear{2013})
\doiurl{10.1142/S0129055X13500177}
\end{barticle}
\endbibitem

%%% 40
\bibitem[\protect\citeauthoryear{Blackadar}{2006}]{Blackadar2006}
\begin{bbook}
\bauthor{\bsnm{Blackadar}, \binits{B.}}:
\bbtitle{Theory of C*-Algebras and Von Neumann Algebras}.
\bpublisher{Springer},
\blocation{Heidelberg}
(\byear{2006})
\end{bbook}
\endbibitem

%%% 41
\bibitem[\protect\citeauthoryear{Robinson}{1982}]{Robinson}
\begin{barticle}
\bauthor{\bsnm{Robinson}, \binits{D.W.}}:
\batitle{Strongly positive semigroups and faithful invariant states}.
\bjtitle{Commun.Math. Phys.}
\bvolume{85},
\bfpage{129}--\blpage{142}
(\byear{1982})
\end{barticle}
\endbibitem

%%% 42
\bibitem[\protect\citeauthoryear{Kraus}{1971}]{kraus71}
\begin{barticle}
\bauthor{\bsnm{Kraus}, \binits{K.}}:
\batitle{General state changes in quantum theory}.
\bjtitle{Ann. Physics}
\bvolume{64},
\bfpage{311}--\blpage{335}
(\byear{1971})
\doiurl{10.1016/0003-4916(71)90108-4}
\end{barticle}
\endbibitem

%%% 43
\bibitem[\protect\citeauthoryear{Roman~Geiko}{}]{GMM}
\begin{botherref}
\oauthor{\bsnm{Roman~Geiko}, \binits{G.M.} \bsuffix{Tom~Mainiero}}:
A Categorical Triality: Matrix Product Factors, Positive Maps, and von Neumann Bimodules
\end{botherref}
\endbibitem

%%% 44
\bibitem[\protect\citeauthoryear{Fannes et~al.}{1994}]{FANNES1994511}
\begin{barticle}
\bauthor{\bsnm{Fannes}, \binits{M.}},
\bauthor{\bsnm{Nachtergaele}, \binits{B.}},
\bauthor{\bsnm{Werner}, \binits{R.F.}}:
\batitle{Finitely correlated pure states}.
\bjtitle{Journal of Functional Analysis}
\bvolume{120}(\bissue{2}),
\bfpage{511}--\blpage{534}
(\byear{1994})
\doiurl{10.1006/jfan.1994.1041}
\end{barticle}
\endbibitem

%%% 45
\bibitem[\protect\citeauthoryear{Verstraete et~al.}{2005}]{Verstraete_2005}
\begin{botherref}
\oauthor{\bsnm{Verstraete}, \binits{F.}},
\oauthor{\bsnm{Cirac}, \binits{J.I.}},
\oauthor{\bsnm{Latorre}, \binits{J.I.}},
\oauthor{\bsnm{Rico}, \binits{E.}},
\oauthor{\bsnm{Wolf}, \binits{M.M.}}:
Renormalization-group transformations on quantum states.
Physical Review Letters
\textbf{94}(14)
(2005)
\doiurl{10.1103/physrevlett.94.140601}
\end{botherref}
\endbibitem

%%% 46
\bibitem[\protect\citeauthoryear{Raeburn and Williams}{1998}]{Raeburn1998MoritaEA}
\begin{bbook}
\bauthor{\bsnm{Raeburn}, \binits{I.}},
\bauthor{\bsnm{Williams}, \binits{D.P.}}:
\bbtitle{Morita Equivalence and Continuous-Trace {C}$^*$-Algebras},
(\byear{1998})
\end{bbook}
\endbibitem

%%% 47
\bibitem[\protect\citeauthoryear{Simon}{1983}]{Simon1983}
\begin{barticle}
\bauthor{\bsnm{Simon}, \binits{B.}}:
\batitle{Holonomy, the quantum adiabatic theorem, and berry's phase}.
\bjtitle{Phys. Rev. Lett.}
\bvolume{51},
\bfpage{2167}--\blpage{2170}
(\byear{1983})
\doiurl{10.1103/PhysRevLett.51.2167}
\end{barticle}
\endbibitem

%%% 48
\bibitem[\protect\citeauthoryear{Berry}{1984}]{Berry}
\begin{barticle}
\bauthor{\bsnm{Berry}, \binits{M.V.}}:
\batitle{Quantal phase factors accompanying adiabatic changes}.
\bjtitle{Proceedings of the Royal Society of London. Series A, Mathematical and Physical Sciences}
\bvolume{392}(\bissue{1802}),
\bfpage{45}--\blpage{57}
(\byear{1984})
\end{barticle}
\endbibitem

%%% 49
\bibitem[\protect\citeauthoryear{Hastings}{2007}]{hastings2007area}
\begin{barticle}
\bauthor{\bsnm{Hastings}, \binits{M.B.}}:
\batitle{An area law for one-dimensional quantum systems}.
\bjtitle{Journal of statistical mechanics: theory and experiment}
\bvolume{2007}(\bissue{08}),
\bfpage{08024}
(\byear{2007})
\end{barticle}
\endbibitem

%%% 50
\bibitem[\protect\citeauthoryear{Serre}{1957}]{Serre}
\begin{barticle}
\bauthor{\bsnm{Serre}, \binits{J.-P.}}:
\batitle{Modules projectifs et espaces fibrés à fibre vectorielle.}
\bjtitle{Séminaire Dubreil. Algèbre et théorie des nombres}
\bvolume{11}(\bissue{09}),
\bfpage{18}
(\byear{1957})
\end{barticle}
\endbibitem

%%% 51
\bibitem[\protect\citeauthoryear{Atiyah and Segal}{2004}]{atiyah2004twisted}
\begin{botherref}
\oauthor{\bsnm{Atiyah}, \binits{M.}},
\oauthor{\bsnm{Segal}, \binits{G.}}:
Twisted {K}-theory
(2004)
\end{botherref}
\endbibitem

%%% 52
\bibitem[\protect\citeauthoryear{Dixmier and Douady}{1963}]{DD}
\begin{barticle}
\bauthor{\bsnm{Dixmier}, \binits{J.}},
\bauthor{\bsnm{Douady}, \binits{A.}}:
\batitle{Champs continus d'espaces hilbertiens et de {C}$^*$-alg\`ebres}.
\bjtitle{Bulletin de la Soci\'et\'e Math\'ematique de France}
\bvolume{91},
\bfpage{227}--\blpage{284}
(\byear{1963})
\doiurl{10.24033/bsmf.1596}
\end{barticle}
\endbibitem

%%% 53
\bibitem[\protect\citeauthoryear{Brylinski}{1993}]{Brylinski}
\begin{bbook}
\bauthor{\bsnm{Brylinski}, \binits{J.-L.}}:
\bbtitle{Loop Spaces, Characteristic Classes and Geometric Quantization}.
\bpublisher{Birkhäuser Boston, MA},
\blocation{Boston}
(\byear{1993})
\end{bbook}
\endbibitem

%%% 54
\bibitem[\protect\citeauthoryear{Becker et~al.}{2006}]{Becker_Becker_Schwarz_2006}
\begin{bbook}
\bauthor{\bsnm{Becker}, \binits{K.}},
\bauthor{\bsnm{Becker}, \binits{M.}},
\bauthor{\bsnm{Schwarz}, \binits{J.H.}}:
\bbtitle{String Theory and M-Theory: A Modern Introduction}.
\bpublisher{Cambridge University Press},
\blocation{Cambridge}
(\byear{2006})
\end{bbook}
\endbibitem

%%% 55
\bibitem[\protect\citeauthoryear{Bunke et~al.}{2006}]{BUNKE_2006}
\begin{barticle}
\bauthor{\bsnm{Bunke}, \binits{U.}},
\bauthor{\bsnm{Rumpf}, \binits{P.}},
\bauthor{\bsnm{Schick}, \binits{T.}}:
\batitle{The topology of {T}-duality for {T}$^n$-bundles}.
\bjtitle{Reviews in Mathematical Physics}
\bvolume{18}(\bissue{10}),
\bfpage{1103}--\blpage{1154}
(\byear{2006})
\doiurl{10.1142/s0129055x06002875}
\end{barticle}
\endbibitem

%%% 56
\bibitem[\protect\citeauthoryear{Rosenberg}{1989}]{Rosenberg_1989}
\begin{barticle}
\bauthor{\bsnm{Rosenberg}, \binits{J.}}:
\batitle{Continuous-trace algebras from the bundle theoretic point of view}.
\bjtitle{Journal of the Australian Mathematical Society. Series A. Pure Mathematics and Statistics}
\bvolume{47}(\bissue{3}),
\bfpage{368}--\blpage{381}
(\byear{1989})
\end{barticle}
\endbibitem

%%% 57
\bibitem[\protect\citeauthoryear{I.~Gelfand}{1943}]{GelfandNeumark1943}
\begin{botherref}
\oauthor{\bsnm{I.~Gelfand}, \binits{M.N.}}:
On the imbedding of normed rings into the ring of operators in {H}ilbert space.
Matem. Sbornik
\textbf{12}(54)
(1943)
\end{botherref}
\endbibitem

%%% 58
\bibitem[\protect\citeauthoryear{Negrepontis}{1971}]{NEGREPONTIS1971228}
\begin{barticle}
\bauthor{\bsnm{Negrepontis}, \binits{J.W.}}:
\batitle{Duality in analysis from the point of view of triples}.
\bjtitle{Journal of Algebra}
\bvolume{19}(\bissue{2}),
\bfpage{228}--\blpage{253}
(\byear{1971})
\doiurl{10.1016/0021-8693(71)90105-0}
\end{barticle}
\endbibitem

%%% 59
\bibitem[\protect\citeauthoryear{Raeburn and Rosenberg}{1988}]{RaeburnRosenberg1988}
\begin{barticle}
\bauthor{\bsnm{Raeburn}, \binits{I.}},
\bauthor{\bsnm{Rosenberg}, \binits{J.}}:
\batitle{Crossed products of continuous-trace {C}$^*$ -algebras by smooth actions}.
\bjtitle{Transactions of The American Mathematical Society - TRANS AMER MATH SOC}
\bvolume{305},
\bfpage{1}--\blpage{1}
(\byear{1988})
\doiurl{10.1090/S0002-9947-1988-0920145-6}
\end{barticle}
\endbibitem

%%% 60
\bibitem[\protect\citeauthoryear{Lee}{2010}]{LeeTop2010}
\begin{bbook}
\bauthor{\bsnm{Lee}, \binits{J.M.}}:
\bbtitle{Introduction to Topological Manifolds}.
\bpublisher{Springer},
\blocation{Boston}
(\byear{2010})
\end{bbook}
\endbibitem

%%% 61
\bibitem[\protect\citeauthoryear{Kirby and Siebenmann}{1969}]{Kirby1969OnTT}
\begin{barticle}
\bauthor{\bsnm{Kirby}, \binits{R.C.}},
\bauthor{\bsnm{Siebenmann}, \binits{L.C.}}:
\batitle{On the triangulation of manifolds and the hauptvermutung}.
\bjtitle{Bulletin of the American Mathematical Society}
\bvolume{75},
\bfpage{742}--\blpage{749}
(\byear{1969})
\end{barticle}
\endbibitem

%%% 62
\bibitem[\protect\citeauthoryear{Rosenberg and Sciences}{}]{rosenbergtopology}
\begin{botherref}
\oauthor{\bsnm{Rosenberg}, \binits{J.}},
\oauthor{\bsnm{Sciences}, \binits{C.B.M.}}:
Topology, {C}$^*$-algebras, and String Duality.
Regional conference series in mathematics.
American Mathematical Soc.
\end{botherref}
\endbibitem

\end{thebibliography}

\end{document}